\documentclass[conference,10pt]{llncs}
\usepackage{amsmath}
\usepackage{times}
\usepackage{algorithm2e}
\usepackage{algorithmicx}
\usepackage{graphicx}
\usepackage{qtree}

\newtheorem{myt}{Theorem}
\newtheorem{myp}{Proposition}
\newtheorem{myl}{Lemma}
\newtheorem{myc}{Corollary}
\newtheorem{myo}{Observation}
\newtheorem{myd}{Definition}

\begin{document}
\title{Total stretch minimization on single and identical parallel machines}

\author{ Abhinav Srivastav\inst{1,2}, Denis Trystram\inst{1.3}}
\institute{Univ. Grenoble Alpes \inst{1}, CNRS-Verimag\inst{2}  \& Institut Universitaire de France \inst{3}
\\ \email{abhinav.srivastav@imag.fr, trystram@imag.fr}
}

\maketitle

\begin{abstract}
We consider the classical problem of scheduling $n$ jobs with release dates on both single and identical parallel machines. 
We measure the quality of service provided to each job by its stretch, which is defined as the ratio of its response time to processing time. 
Our objective is to schedule these jobs non-preemptively so as to minimize total stretch.
So far, there have been very few results for total stretch minimization especially for the non-preemptive case. 
For the preemptive version, the Shortest remaining processing time (SRPT) algorithm is known to give $2$-competitive ratio for total stretch on single machine while it has $13$-competitive ratio on identical parallel machines. 
We study the problem with some additional assumptions and present the stronger competitive ratio.
We show that the Shortest processing time (SPT) algorithm is $(\Delta - \frac{1}{\Delta}+1$)-competitive for non-preemptive total stretch minimization on single machine and it is $(\Delta -\frac{1}{\Delta}+ \frac{3}{2} -\frac{1}{2m})$ on $m$ identical parallel machines, where $\Delta$ is the upper bound on the ratio between the maximum and the minimum processing time of the jobs. 
\end{abstract}
%
%
%
\section{Introduction}

\par We consider the problem of non-preemptive scheduling of jobs with release dates on single and identical parallel machines. 
Our objective is to schedule these jobs so as to guarantee the \textit{``fair"} quality of service to individual jobs. 
Stretch is defined as a factor by which a job is slowed down with respect to the time it takes on unloaded system~\cite{Bender98flowand}. 
Formally, we are given a set of $n$ jobs where the job $J_j$ has a processing time $p_j$ and a release date $r_j$ before which it cannot be scheduled, then the stretch $s_j$ of job $J_j$ is formally defined as $\frac{F_j}{p_j}$, where $F_j = C_j - r_j$ denotes the flow time ($C_j$ being the completion time of job $J_j$ in the schedule). 
Our objective is to schedule the stream of jobs arriving online so as to minimize $\sum s_j$ for all the instances.
This objective is often referred to as the \textit{average stretch} or \textit{total stretch} optimization problem.
In this paper, we restrict our attention to schedule the jobs non-preemptively on single and parallel machines.
%
%
In the classical scheduling notation introducted by Graham et al.~\cite{Graham:1979}, these problems are respectively represented as $1|r_j|\sum s_j$ and $Pm|r_j|\sum s_j$. 
Legrand et al.~\cite{Arnaud2008} showed using reduction from \emph{partition} problem, that $1|r_i|\sum s_i$ is NP complete. 
\subsection{Related works}
\par Muthukrishnan et al.~\cite{min-average-stretch} showed that the classical scheduling policy, Shortest Remaining processing time (SRPT) is $2$ and $13$-competitive for the problem of $1|r_j, pmtn| \sum s_j$ and $Pm|r_j, pmtn| \sum s_j$, respectively.
Later, Chekuri et al.~\cite{Chekuri:2001} presented an algorithm that achieves competitive ratios of $13$ and $19$ for migratory and non-migratory models of $Pm|r_j, pmtn| \sum s_j$.
Bender et al.~\cite{Bender03} presented a PTAS for uniprocessor preemptive case of total stretch with running time in $O(n^{poly(\frac{1}{\epsilon})})$. 
A more general problem than the total stretch is the problem of minimizing the sum of weighted flow time ($\sum w_i F_i$).
There is no online algorithm known with the constant competitive ratio for the sum of weighted flow time.
Bansal et al.~\cite{Nikhil:2007} showed using resource augmentation that there is an $O(1)$-speed $O(1)$-approximation for the offline version of weighted sum flow problem. 
Leonardi et al.~\cite{Leonardi:2007} proved the lower bound of $\Omega(n^{\frac{1}{2} -\epsilon})$ for $1|r_j|\sum F_j$, while Kellerer et al.~\cite{Kellerer95approximability} showed that the worst case has a lower bound of $\Omega(n^{\frac{1}{3}-\epsilon})$ for $Pm|r_j|\sum F_j$. 
Considering that such strong lower bounds exist for sum flow time, we assume additional information that the ratio of maximum processing time over minimum processing time for all the jobs is bounded by $\Delta$. 
Using this assumption Bunde~\cite{Bunde04sptis} proved that the Shortest Processing time (SPT) algorithm is $\frac{\Delta+1}{2}$-competitive for the sum flow time on a single machine. 
Chekuri et al.~\cite{Chekuri:2001} provided an online algorithm for $1|r_i,pmtn|\sum_i  w_i F_i$ that is $O(log^2\Delta)$-competitive. 
They also give a quasi-polynomial time $(2+\epsilon)$-approximation for the offline case when the weights and processing times are polynomial bounded. 
Tao et al.~\cite{Tao:2013} showed that Weighted shortest processing time is $\Delta+1$ and $\Delta + \frac{3}{2} - \frac{1}{2m}$-competitive for sum of weighted flow time on single and parallel machines, respectively.
Their analysis is based on the idea of instance transformation which inherently assumes that the weights are independent of job's parameters.  For the case of stretch minimization, this assumption is not valid.
We provide proof for this special case where weights are dependent on processing times i.e $w_i = \frac{1}{p_i}$. Moreover, the competitive ratios presented in the paper, are tighter in comparison to that of Tao et al.~\cite{Tao:2013}.
\subsection{Contributions}
\par In this paper, we extend the understanding of the competitiveness of stretch for non-preemptive schedules by presenting new competitive ratios. 
We show that SPT provides $(\Delta - \frac{1}{\Delta}+1)$ and $(\Delta -\frac{1}{\Delta} + \frac{3}{2} - \frac{1}{2m})$-competitiveness in $O(n\log n)$ time for problem of $1|r_j|\sum s_j$ and $Pm|r_j|\sum s_j$, where $m$ is the number of machines. 
%
%
Our analysis for single machine is based on careful observations on the structural similarity between SPT and SRPT schedules. 
On another hand, our analysis for parallel machine is based on converting SPT on parallel machines to a new schedule on a virtual single machine. 
\par This paper is organised as follows. In Section \ref{pre}, we present basic definitions and notations used in this paper. 
In Section \ref{1machine}, we analyze the SPT algorithm on a single machine while in section~\ref{m-machine}, we present the analysis of SPT on identical parallel machines. 
Section~\ref{conc} provides some concluding remarks for this work.
%
%
%
\section{Preliminaries} \label{pre}

\par In this section, we introduce some basic definitions and notations, that are used frequently in the remainder of this paper. 
We consider the following clairvoyant online scheduling scenario. 
A sequence of jobs arrive over time and the processing time of each job is known at its time of arrival. 
Our goal is to execute the continuously arriving stream of jobs. Let $\mathcal{I}$ be a given scheduling instance specified by a set of jobs $\mathit{J}$, and for each job $J_j\in\mathit{J}$, a release time $r_j$ and a processing time $p_j$. 
Without loss of generality, we assume that the smallest and largest processing times are equal to 1 and $\Delta$, respectively. 
%
%
%
%
\par The proposed work is focused on studying two different well-known algorithms, namely SRPT and SPT. 
The Shortest Remaining Processing Time (abbreviated as SRPT) is a preemptive schedule which can be defined as follows: 
at any time $t$, the available job $J_j$ with the shortest remaining processing time $\rho_j(t)$ is processed until it is either completed or until another job $J_i$ with $\rho_i< \rho_j(r_i)$ 
becomes available, where the remaining processing time $\rho_j(t)$ of job $J_j$ is the amount of processing time of $J_j$ which has not been scheduled before time $t$. 
In the second case, job $J_j$ is preempted and job $J_i$ is processed. 
On another hand, the Shortest Processing Time  (abbreviated as SPT) is a non-preemptive non-waiting schedule that runs the shortest available job in the queue whenever a processor becomes idle.
\par Formally, an online algorithm $\mathcal{A}_{on}$ is said to be $\alpha$-competitive with respect to a offline algorithm $\mathcal{A}_{off}$ if the worst case ratio (over all possible instances) of the performance of $\mathcal{A}_{on}$ is no more that $\alpha$ times the performance of $\mathcal{A}_{off}$. 
%
\section{Analysis of $1|r_i|\sum s_i$} \label{1machine}
\par We begin by introducing some notions of schedule that play a central role in our analysis. We first show in section \ref{srpt-spt}, the structural similarity between SRPT and SPT schedules. Then, we construct a non-preemptive schedule by changing SRPT into a new schedule (called POS) and show that POS is $(\Delta - \frac{1}{\Delta}+1)$-competitive for SRPT (section~\ref{inters}). 
Later, we show in section~\ref{spt-ana} that the total stretch of SPT and SRPT are no worse than that of POS and non-preemptive schedules, respectively. 
Thus, the cost of SPT is $(\Delta-\frac{1}{\Delta}+1)$ factor within the cost of an optimal non-preemptive schedule.
\subsection{Structure of SRPT and SPT} \label{srpt-spt}
\par Without the loss of generality (W.l.o.g), we assume that SRPT resumes one of the jobs with equal remaining processing time before executing a new job. 
Though it may choose arbitrarily between jobs with equal initial processing times, provided that SPT uses the same order. 
In SRPT we define an active interval $I_j = [S_j,C_j]$ for each job $J_j$, where $S_j$ is the start time of $J_j$ in the preemptive schedule. 
Note that due to the preemptive nature of schedule, the length of $I_j$ (denoted by $|I_j|$) is greater than or equal to the processing time $p_j$ of job $J_j$. 
When two such active intervals intersect, one is contained in the other and there is no machine idle time in between both intervals~\cite{Kellerer95approximability}.
\par Based on such strong containment relations, we define a \textit{directed ordered forest} as shown in Figure~\ref{srpt-dof}. The vertices consist of jobs $1,....,n$. 
There exists a directed edge going from job $J_i$ to $J_j$ if and only if $I_j \subseteq I_i$ and there does not exist a job $J_k$ with $I_j \subseteq I_k\subseteq I_i$. 
For every vertex $i$, its children are ordered from left to right according to the ordering of their corresponding intervals in $I_i$. 
Hence, we have a collection of directed out-trees $\mathcal{T} = \{T_1,....,T_r\}$. 
We also order the roots $\gamma(T_k)$ of trees  from left to right according to the ordering of their corresponding intervals. 
Hence, SRPT runs the jobs in order of out-trees that is: 
for every out-tree, all the jobs belonging to an out-tree $T_a$ are executed before $T_b$ if and only if $I_{\gamma(T_a)}<I_{\gamma(T_b)}$.  
Bunde showed in~\cite{Bunde04sptis} that SPT also runs the job in similar fashion. 
Thus, the difference between SRPT and SPT comes from the order of execution of jobs within each out-tree.
\begin{figure} 
\centering
\includegraphics[width=4in]{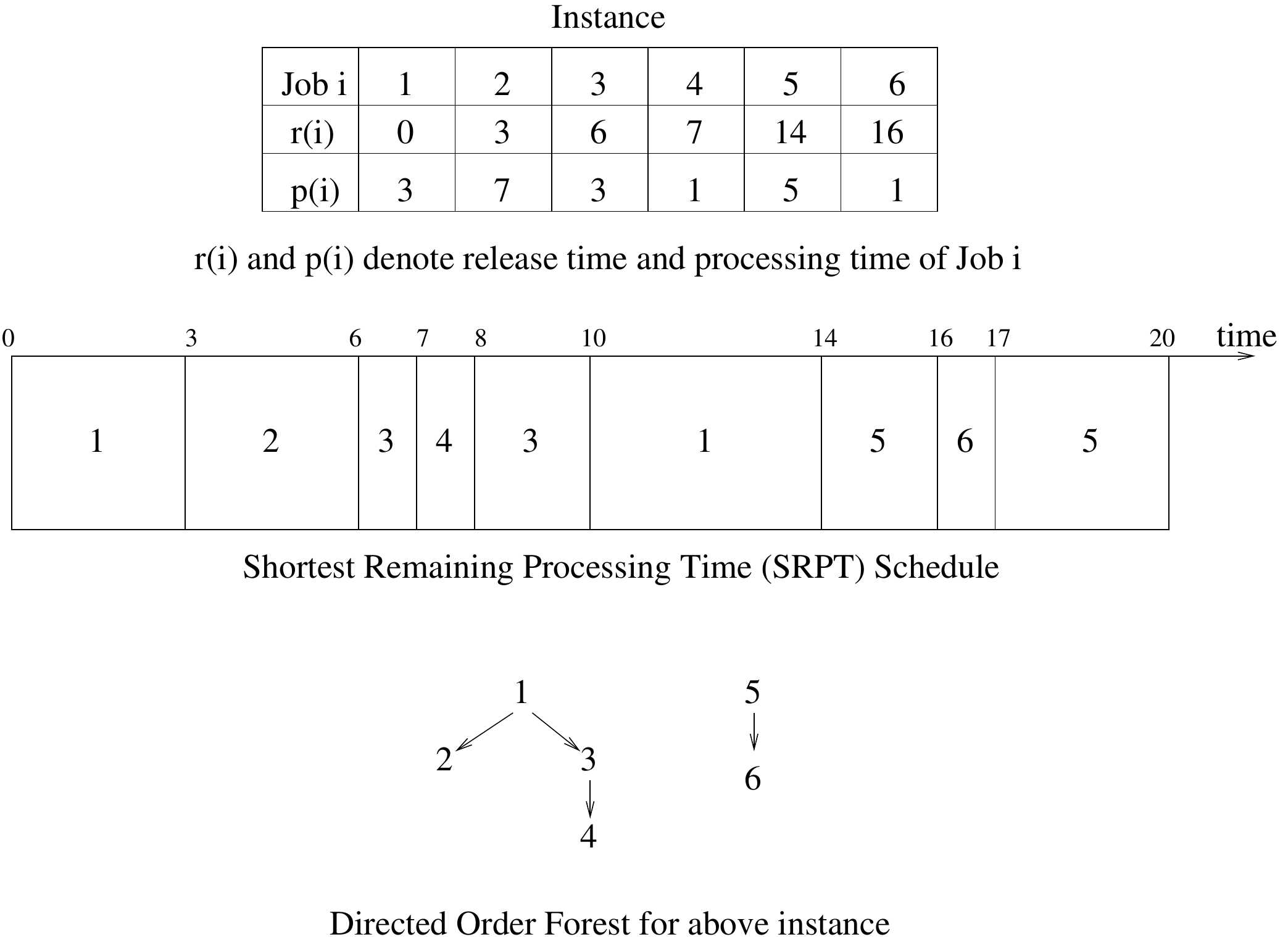}
\caption{Example showing a SRPT schedule and its corresponding directed ordered forest}
\label{srpt-dof}
\end{figure}
\subsection{Intermediate schedule} \label{inters}
\par Starting from SRPT schedule, we construct a new non-preemptive schedule called \textit{POS} (which stands for Post Order Schedule). 
During the interval $I_j = [S_{j},C_{j}]$, where $j = \gamma(T_a)$, POS runs the jobs of $T_a$, starting with $J_j$ and then running the other jobs of $T_a$ 
in order of increasing SRPT completion time (post order transversal) as shown in Figure~\ref{srpt-2-pos}.
\begin{figure}
\centering
\includegraphics[width=3.5in]{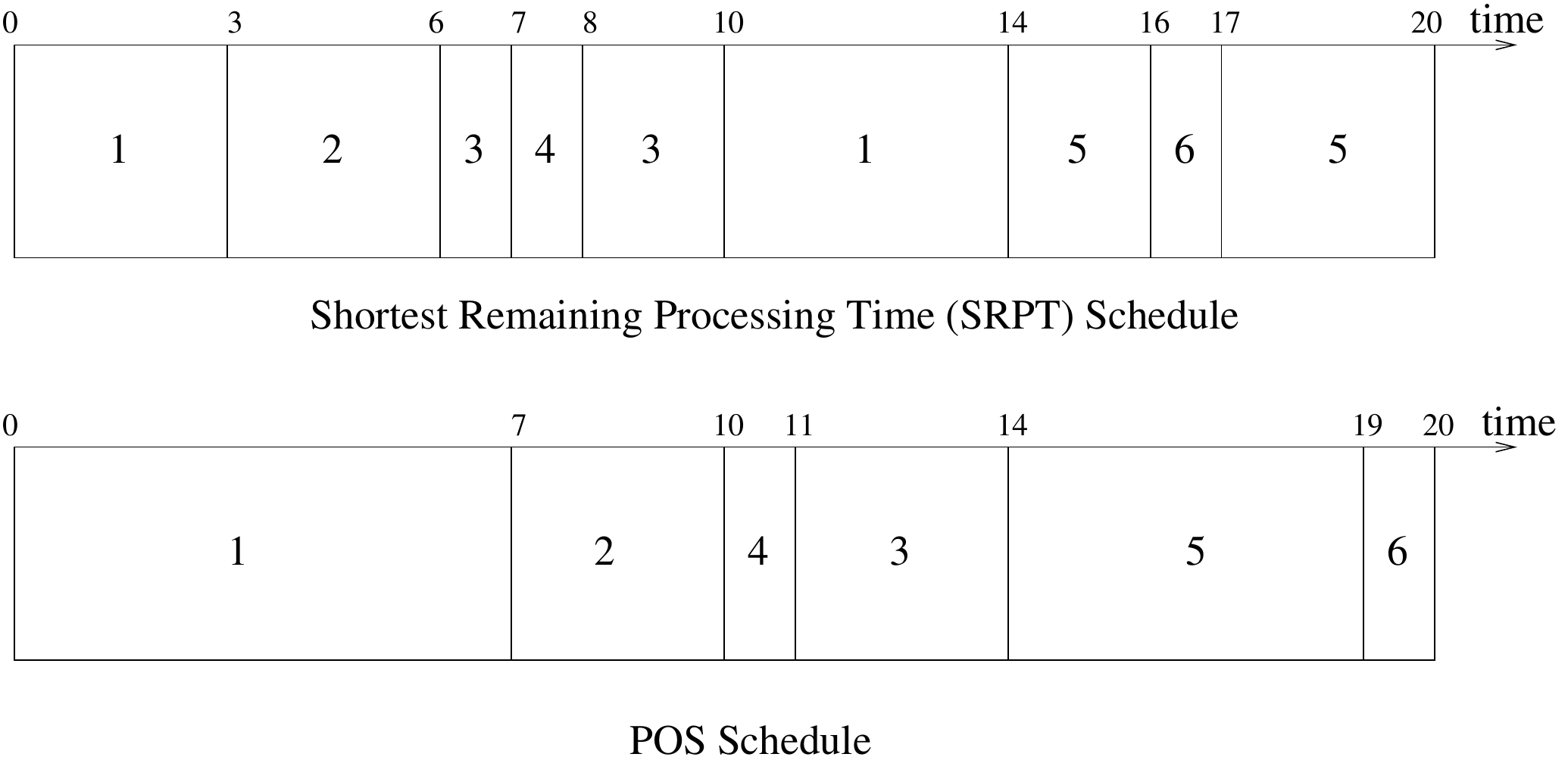}
\caption{Transformation showing SRPT to POS schedule}
\label{srpt-2-pos}
\end{figure}
\begin{myd} \label{def-compact}
A schedule is said compact if there is no idle time between the execution of jobs except due to the unavailability of jobs.
\end{myd}
%
%
%
%
\begin{myp} \label{pos-compact}
POS is compact if and only if SRPT is compact
\end{myp}
\begin{proof}
There can be idle time in SRPT schedule only when there is no job available for the execution. 
During the execution of jobs belonging to the out-tree $T_a$ in SRPT, at any time $t \in  [S_{T_{\gamma(T_a)}},C_{T_{\gamma(T_a)}}]$, there is at least one partial uncompleted job (specifically, $\gamma(T_a)$) available for the execution. 
Thus, idle time can only exist between intervals of out-tree $T_a$ and $T_{a+1}$. 
It is sufficient to show that POS completes all the jobs belonging to $T_a$ in interval $I_{\gamma(T_a)}$.
\par At any time $t\geq S_{\gamma(T_a)}+p_{\gamma(T_a)}$, the job scheduled in POS, has already been completed by SRPT. 
This follows from the fact that SRPT has always finished at least as many jobs as in any other algorithm~\cite{Chung:2010}. 
While at any time $t < S_{\gamma(T_a)} + p_{\gamma(T_a)}$, POS runs the root job of $T_a$ i.e $\gamma(T_a)$.
Hence, POS is busy for entirety of the interval $I_{\gamma(T_a)}$ completing all the jobs belonging to $T_a$.
\end{proof}
The following corollary is a direct consequence of proposition~\ref{pos-compact}.

\begin{myc} \label{pos-srpt-order}
POS and SRPT only differ in the order of execution of the jobs within each out-tree.
\end{myc}
\begin{myl} \label{pos-comp}
POS is $(\Delta-\frac{1}{\Delta}+1)$-competitive for total stretch with respect to SRPT schedule.
\end{myl}
\begin{proof}
From corollary~\ref{pos-srpt-order}, it is sufficient to show that POS schedule is $(\Delta-\frac{1}{\Delta}+1)$-competitive with respect to SRPT schedule for any out-tree. 
Let $\sum\limits^{srpt}$ and $\sum\limits^{pos}$ denote the sum stretch of SRPT and POS for an out-tree $T_a$, respectively. 
The completion time of $\gamma(T_a)$ in POS, is $\sum\limits_{\forall J_k \in T_a}^{J_k \neq \gamma(T_a)} p_k$ time units earlier than that of SRPT. 
Consequently, $s_{\gamma(T_a)}^{pos} = s_{\gamma(T_a)}^{srpt} - \sum\limits_{\forall J_k \in T_a}^{J_k \neq \gamma(T_a)} \frac{p_k}{p_{\gamma(T_a)}}$.  
On another hand, the rest of the jobs are delayed in POS by at most $p_{\gamma(T_a)}$. 
Therefore, $s_k^{pos}\leq s_k^{srpt} + \frac{p_{\gamma(T_a)}}{p_k}, \forall J_k \in T_a, J_k \neq \gamma(T_a) $. Then, the sum stretch of all jobs in $T_a$ is given by :
\begin{align*}
\sum_{k\in T_a} s_k^{pos} &\leq \sum_{k\in T_a} s_k^{srpt} + \sum_{\forall J_k \in T_a}^{J_k \neq \gamma(T_a)}\frac{p_{\gamma(T_a)}}{p_k} -  \sum\limits_{\forall J_k \in T_a}^{J_k \neq \gamma(T_a)} \frac{p_k}{p_{\gamma(T_a)}}  \\
\end{align*}
\begin{align*}
\sum\limits^{pos} &\leq \sum\limits^{srpt} +  \sum_{\forall J_k \in T_a}^{J_k \neq \gamma(T_a)}(\frac{p_{\gamma(T_a)}}{p_k} - \frac{p_k}{p_{\gamma(T_a)}}) \\
&\leq \sum\limits^{srpt} + \sum_{\forall J_k \in T_a}^{J_k \neq \gamma(T_a)} (\Delta - \frac{1}{\Delta}) \\
&\leq \sum\limits^{srpt} + ||I_a|| (\Delta - \frac{1}{\Delta}) \leq (\Delta - \frac{1}{\Delta} +1)\sum\limits^{srpt} 
\end{align*}
where $||I_a||$ denotes the number of jobs executed in interval $I_a$. The third inequality is direct consequence of fact that $1 \leq \frac{p_{\gamma(T_a)}}{p_k} \leq \Delta$. The last inequality follows from the fact that $s_j \geq 1$, $\forall j \in J$.
\end{proof}
\subsection{SPT versus POS} \label{spt-ana}
\par In this section, we show that SPT achieves lower total stretch than that of POS. Our proof is based on iteratively changing SPT schedule to POS by removing the first difference between them. 
\par W.l.o.g, assume that SPT runs the jobs in numerical order: $J_1$ followed by $J_2$ and so on. Let the first difference between SPT and POS occurs when SPT starts a jobs $J_i$ while POS starts another job $J_j$ as shown in Figure~\ref{spt2pos}. 
SPT is changed by moving $J_j$ before $J_i$ and shifting every job from $i$ to $j-1$. Hence, the increase in the stretch (denoted by $\delta_{j}$) by the above transformation, is given by
\begin{align} \label{eq1}
\delta_{j} &= \sum\limits_{k=i}^{j-1}\frac{p_j}{p_k} - \sum\limits_{k=i}^{j-1}\frac{p_k}{p_j} =  \sum\limits_{k=i}^{j-1}(\frac{p_j}{p_k} - \frac{p_k}{p_j}) = \sum\limits_{k=i}^{j-1}\delta_{jk} \nonumber
\end{align} 
where $\delta_{jk}$ is the local increase in stretch by swapping job $J_k$ $(i \leq k \leq j-1)$ with $J_j$.
\begin{figure}
\includegraphics[scale=0.5]{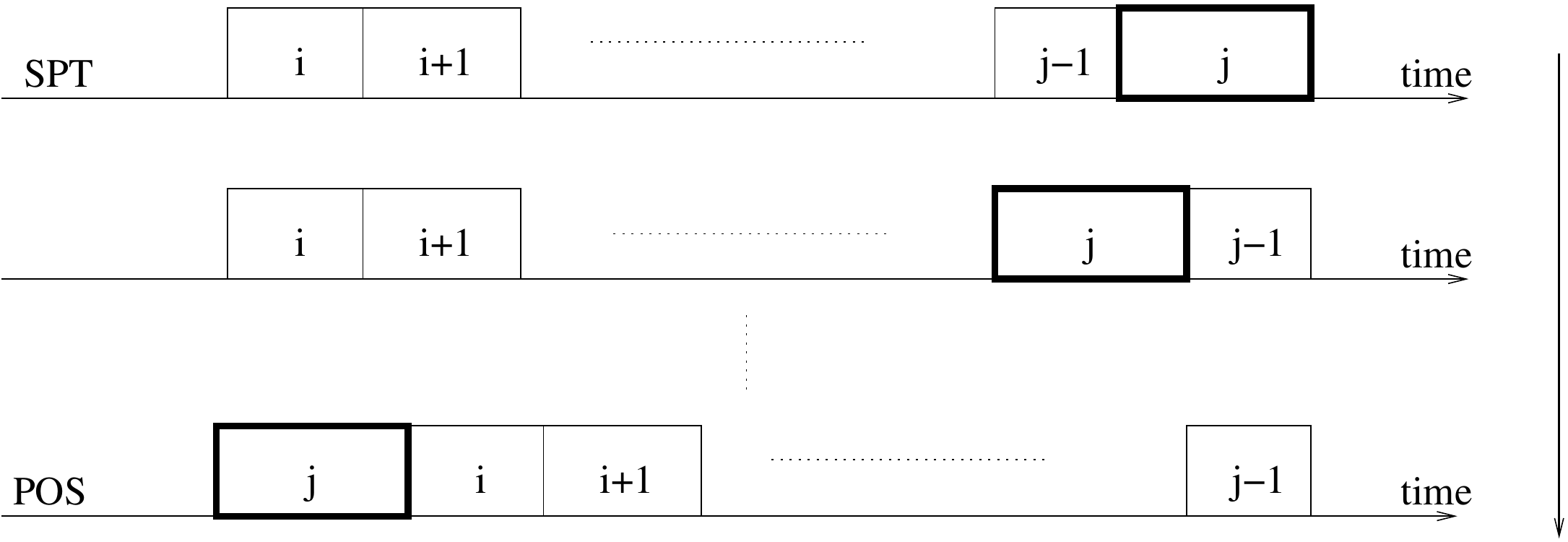}
\caption{Successive transformations from SPT to POS}
\label{spt2pos}
\end{figure}
%
%
Based on sizes of $J_k$ and $J_j$, we classify the transformation moves into two sets according to the sign of $\delta_{jk}$, namely  
$\delta_{jk} <0$ if $p_j < p_k$ and $\delta_{jk} \geq 0$ otherwise. Now, we present a series of technical results for proving that $\sum \delta_{j} \geq 0$, $\forall j \in J$. \\
\par Let consider two jobs $J_j$ and $J_k$ such that $\delta_{jk} < 0$; $J_j \prec J_k$ in POS and $J_k \prec J_j$ in SPT schedule, where $a \prec b$ denotes that job $a$ is executed before job $b$ in the schedule.
 
%
%
\begin{myl}  \label{l>j}
There exists a job $J_l$ ($l\neq j$) such that $J_l \prec J_k$ in POS while $J_k \prec J_l$ in SPT and $p_{l} \geq p_k >p_j$.
\end{myl}
\begin{proof}
Let $t'$ be the time at which SPT schedules job $J_k$. Since $p_j < p_k$ and $J_k \prec J_j$ in POS, it follows that $t' < r_j$. Let $t \geq r_j >t'$  be the time at which POS start executing $J_k$. Therefore at time $t'$, POS schedules some job $J_l$ while SPT schedules $J_k$. Now if $p_l \geq p_k$, our lemma holds. Otherwise if $p_l < p_k$, then using the argument iteratively, it can be shown that there is some other job $J_{l'}$ such that $J_k \prec J_{l'}$ in SPT while $J_{l'} \prec J_k$ in POS. 
%
%
%
%
%
%
\end{proof}
\begin{myo} \label{order-pos}
$J_l \prec J_j \prec J_k$ in POS schedule.
\end{myo}
\begin{myl} \label{J_j,J_l}
$J_j \prec J_l$ in SPT. 
\end{myl}
\begin{proof}
Assume that in SPT $J_l \prec J_j$. 
Let $t$ and $t'$ be the time at which $J_l$ start executing in SPT and POS schedules, respectively.
Then using Lemma~\ref{l>j}, $t' < t$.
Using SPT principle of the schedule at $t$, we get $r_j > t$ and $p_l > p_j$ in which case POS cannot schedule $J_l$ at time $t'$ or $p_l \leq p_j$.
Hence, this yields a direct contradiction to our assumption that $J_l \prec J_j$.

%
%
\end{proof}
\begin{myc} \label{sum>0}
For each $\delta_{jk} <0$ , there exists a job $l$ such that $\delta_{jk} + \delta_{lj} \geq 0$.  
\end{myc}
\begin{proof}
It follows from Lemma~\ref{l>j} and Lemma~\ref{J_j,J_l} that there is always a job $l$ such that the ordering of jobs in SPT is $J_k \prec J_j \prec J_l$ while the same set of jobs are executed in order $J_l \prec J_j \prec J_k$ in POS. Moreover, 
the transformation of $\delta_{lj}$ can be coupled with $\delta_{jk}$ such that:
\begin{align} \nonumber 
\delta_{lj} + \delta_{jk} &= \frac{p_l}{p_j} - \frac{p_j}{p_l} + \frac{p_j}{p_k} - \frac{p_k}{p_l} 
= \frac{p_k p_l + p_j^2}{p_j p_k p_l} (p_l - p_k) \geq 0 \nonumber
\end{align}
The last inequality follows from Lemma~\ref{l>j}.
\end{proof}
\begin{myp} \label{all>0}
 $\sum \delta_{j} \geq 0$
\end{myp}
\begin{proof}
It follows from corollary~\ref{sum>0} that for every $\delta_{jk} < 0$, there exists a job $l$ such that $\delta_{jk}
+ \delta_{lj} \geq 0$. Thus, we say that job $k$ is matched to job $l$. Our lemma holds if there exists an injective mapping for all jobs $k_i$ who $\delta_{jk_i} \leq 0$. In case if there are jobs that are surjectively mapped to the same job $l$. Then using a similar proof construction as given on Lemma~\ref{l>j}, where $p_k$ can be replaced with $\sum p_{k_i}$, we get that $p_l \geq \sum\limits_{\forall k_i \in \mathcal{K}} p_{k_i}$.  Consequently, the decrease in the total stretch due to $k_i$'s can be mapped with $\delta_{lj}$ such that $\delta_{lj} + \sum\limits_{\forall k_in \in \mathcal{K}} \delta_{jk_i} \geq 0$. 
%
\end{proof}
The following result is the immediate consequence of Proposition~\ref{all>0} and Lemma~\ref{pos-comp}.
\begin{myc} \label{spt-approx-pos}
SPT is a $(\Delta - \frac{1}{\Delta} +1)$-approximation with respect to SRPT. 
\end{myc}
\subsection{SRPT versus Optimal offline schedule}
\begin{myl} \label{srpt-as-lower-bound}
The total stretch of SRPT schedule is no worse than that of an optimal offline non-preemptive schedule.
\end{myl}
\begin{proof}
The methodology used in the proof is based on comparison of optimal and SRPT schedules at successive idle times. 
The detailed proof is provided in Appendix~\ref{Proof-srpt-as-lower-bound}. 
\end{proof}
\begin{myt} \label{SPT-competitive}
SPT is $(\Delta - \frac{1}{\Delta}+1)$-competitive with respect to the non-preemptive optimal total stretch. 
\end{myt}
\begin{proof}
The theorem follows from combination of Lemma~\ref{srpt-as-lower-bound} and Corollary \ref{spt-approx-pos}.\
\end{proof}
%
%
%
%
%
%
%
\section{Analysis for $Pm|r_i|\sum s_i$} \label{m-machine}
The construction of the directed ordered forest is no more feasible in case of $m$ machines since the jobs may migrate onto different machines. 
Therefore, we propose in section~\ref{construction-omms} a transformation of SPT on $m$-identical machines to a schedule (called OMMS) on a virtual machine with $m$ times the speed of single machine.
Later, it is shown in section~\ref{omms-sptm} that OMMS schedule is a $(\Delta-\frac{1}{\Delta}+1)$-approximation with respect to SPT on virtual machine (SPTM) with $m$ times speed.
Using the lower bound established by Chou et al.~\cite{Chou:2006} on weighted sum flow problem, finally we show that SPT is $\Delta - \frac{1}{\Delta}+ \frac{3}{2}-\frac{1}{2m}$-competitive. 
\subsection{Intermediate schedule OMMS} \label{construction-omms}
\par To every instance $\mathcal{I}$ of the m-identical parallel machines problem, we associate an instance $\mathcal{I}^m$ with the same job set $J$ and for each job $J_j \in J$ , the processing time of $J_j$ is $p_j^m = p_j/m$. The job release dates $r_j$ are unchanged.  Intuitively, the $m$-identical parallel machines are replaced by a virtual machine with speed $m$ times the speed of single machine. Let $C_j^m$ and $F_j^m$ denote the completion time and flow time (defined as $C_J^m - r_j$) of job $j\in J$ on virtual machine. We now define a general rule of transforming any $m$-identical machine schedule into a feasible schedule on virtual machine. 
\begin{myd}
We construct a feasible schedule (called OMMS schedule) for instance $\mathcal{I}^m$ on a $m$-speed virtual machine by transforming SPT schedule for instance $\mathcal{I}$ on $m$ identical parallel machines. The jobs are executed in OMMS in the increasing order of their starting time in SPT, where ties are broken by executing the jobs in non-decreasing order of processing time. 
\end{myd}
\begin{myl} \label{faster-machine-flow}
\begin{align*}
\sum\limits^{spt}_m \leq \sum\limits^{omms} + (1-\frac{1}{m})n
\end{align*}
where $\sum\limits^{spt}_m$ denotes the total stretch for SPT on $m$-identical parallel machines and $\sum\limits^{omms}$ denotes the total stretch of OMMS on a virtual machine.
\end{myl}
\begin{proof}
While transforming the schedule from SPT to OMMS schedule, the processing time of each job $j$ is reduced from $p_j$ to $p_j/m$. Therefore, the difference between $C_j$  and $C_j^m$ of every job $j \in J$ is upper bound by $(1-\frac{1}{m})p_j$.  Taking summation over all $j \in J$, we obtain $\sum\limits^{spt}_m \leq \sum\limits^{omms} + (1-\frac{1}{m})n$.
\end{proof}
\subsection{Relationship between OMMS and SPT on virtual machine} \label{omms-sptm}
\par Here, we define the \textit{block} structure for OMMS schedule based on \textit{compactness} as defined in Definition~\ref{def-compact}. Our approach consists of partitioning the set of jobs into \textit{blocks} $B(1)$,$B(2)$ and so on, such that jobs belonging to any block can be scheduled regardless of jobs belonging to other blocks. Finally, we show that OMMS is $(\Delta+1-\frac{1}{\Delta})$-approximation to SPT on a $m$-speed virtual machine. 
\par Let $R=\{r(1), r(2), r(3),...r(n')\}$ where $n'\leq n$ be the set of all different release times. Assume w.lo.g that $r_1< r_2 < ,...< r_{n'}$. We partition the jobs according to their release times into set of blocks. Let $Q(i) = \{J_j: r_j = r(i)\}$, $i = 1,...,n'$ denotes the set of jobs released at time $r(i)$.  
The block $B(w)$ is defined as follows:
\begin{align}
B(w) &= \bigcup\limits_{i=b_{w-1}+1,..b_w}Q(i) \nonumber 
\end{align}
where $b_w$ is the smallest positive integer such that 
\begin{align}
r(b_{w-1}+1) + \sum\limits_{i=b_{w-1},....,b_w}\sum\limits_{J_j\in Q(i)}  p_j^m < r(b_w)  \nonumber
\end{align}
Intuitively, all jobs in $B(w)$ can be \textit{compactly} scheduled between $r_{B(w)} = min_{J_j\in B(w)} r_j$  and first job of $B(w+1)$ is released. Hence jobs belonging to the first block $B(1)$ could be completed at most time $r(b_1+1)$.
\par Thus, we focus our attention only to jobs belonging to \emph{single block}. We re-define $\mathcal{I}^m$ to denote the instance of jobs in a block.
Hence, $n$ denotes the number of jobs and $J = \{J_1,J_2,...,J_n\}$ denotes the set of jobs in $\mathcal{I}^m$. 
Our next objective is to replace OMMS in lemma~\ref{faster-machine-flow}.
\begin{myd}
We construct a new schedule (SPTM) by scheduling all the jobs of instance $\mathcal{I}^m$ according to SPT rule on $m$-speed virtual machine. 
\end{myd} 
\begin{myo}\label{omms-sptm-block}
OMMS and SPTM have the same sets of jobs in each block.
\end{myo}
\begin{myd}
We construct a new schedule (D-SPTM) by delaying the start of each job in SPTM by $\Delta - \frac{1}{\Delta} $ time units later. 
\end{myd}
\begin{myl}\label{omms-approx-sptm}
\begin{align*}
\sum\limits^{omms} \leq \sum\limits^{d-sptm} \leq (\Delta - \frac{1}{\Delta} + 1)\sum\limits^{sptm}
\end{align*}
where $\sum\limits^{\eta}$ denotes the total stretch of schedule $\eta \in \{$omms, d-sptm, sptm$\}$.  
\end{myl}
\begin{proof}
The first inequality can be proved by removing the first difference between D-SPTM and OMMS schedules as shown for SPT and POS schedules in section~\ref{spt-ana}. The second inequality is the direct consequence of the construction of D-SPTM from SPTM, where each job's completion time is increased by $\Delta - \frac{1}{\Delta}$ with respect to SPTM. 
\end{proof}
%
%
%
%
\begin{myc} \label{spt-sptm}
Combining the results of lemma~\ref{faster-machine-flow} and lemma~\ref{omms-approx-sptm}, it follows that:
\begin{align}
\sum\limits^{spt}_m \leq (\Delta - \frac{1}{\Delta}+1)\sum\limits^{sptm} + (1 - \frac{1}{m})n   \nonumber
\end{align} 
where $\sum\limits^{sptm}$ denotes the total stretch for SPTM. 
\end{myc}
Next, using the bound established by Chou et al. in~\cite{Chou:2006}, we show the relationship between SPTM and the optimal total stretch on $m$-machines.
They gave a lower bound on the weighted completion time problem of $Pm|r_j|\sum w_j C_j$ in terms of LP schedule~\footnote{which was first defined in work of Goemans et al.~\cite{Goemans02}} on $m$-speed virtual machine. 
They proved that:
\begin{align} 
\sum w_j C_j^* \geq \sum w_j C_j^{LP} + \frac{1}{2}(1-\frac{1}{m}) \sum w_j p_j 
\end{align} 
For non-preemptive scheduling problem, the order of execution of jobs in LP schedule is similar to that of SPT.  
Here, we extend the above result to the total stretch problem due to the equivalence between optimal schedules for $Pm|r_j|\sum w_j C_j$ and $Pm|r_j|\sum s_j$ where for each job $j$, $w_j = \frac{1}{p_j}$.
\begin{myc} \label{opt-lb}
Let $\sum\limits^{opt}_{m}$ denotes the optimal total stretch for $m$-identical parallel machines, then
\begin{align*} 
\sum\limits^{opt}_m &\geq \sum\limits^{sptm} + \frac{1}{2}(1 - \frac{1}{m})n 
\end{align*}
\end{myc}
\begin{myt} \label{final-theorem}
SPT is $(\Delta - \frac{1}{\Delta} + \frac{3}{2} -\frac{1}{2m})$-competitive for total stretch on $m$-identical parallel machine. 
\end{myt}
\begin{proof}
The proof directly follows from using Corollary~\ref{opt-lb} and ~\ref{spt-sptm} along with the fact that $\sum\limits^{opt}_m \geq n$. For details, please refer to Appendix~\ref{final-proof}.
\end{proof}
\section{Concluding remarks} \label{conc}
\par In this paper, we investigated the problem of minimizing total stretch (average stretch) on a single and parallel machines. 
We give tighter bound in comparison to previous work of Tao et al.~\cite{Tao:2013} on weighted sum flow time. 
The main results, obtained by a series of intermediate schedules, show that the well-known scheduling policy SPT achieves competitive ratios of $(\Delta - \frac{1}{\Delta}+1)$ and $\Delta - \frac{1}{\Delta}+ \frac{3}{2} - \frac{1}{2m}$, respectively for the single and parallel machines cases. 
%
%
%
%
\bibliographystyle{unsrt}
\bibliography{sum-stretch}
\newpage
\appendix
\section{Proof of Lemma~\ref{srpt-as-lower-bound}} \label{Proof-srpt-as-lower-bound}
\begin{proof} 
Without the loss of generality, assume that optimal schedule run jobs in numerical order: $J_1$ followed by $J_2$ and so on. 
\begin{itemize}
\item First, we consider the case where the optimal schedule is \emph{compact}. Let $J_k$ and $J_{k+1}$ be two consecutive jobs. Let $t$ be time at which $J_{k}$ starts its execution. Then either 
\begin{itemize}
\item $p_{k+1} > p_{k}$, then \\ 

SRPT will also schedule $J_k$ before $J_{k+1}$ at time $t$. \\

\item $p_{k+1} \leq p_{k}$ and $t < r_{k+1} < t+p_k$, then \\ 

$s_k = \frac{t+p_k-r_k}{p_k}$ and $s_{k+1} = \frac{t+p_k+p_{k+1}-r_{k+1}}{p_{k+1}}$. \\

Let $w = r_{k+1} - t$ . Then, the following inequality holds since the optimal schedule is compact.
\begin{align}
s_k+s_{k+1} &\leq \frac{t+w+p_k+p_{k+1}-r_k}{p_k} + \frac{t+w+p_{k+1}-r_{k+1}}{p_{k+1}} \nonumber \\
w &\geq p_k-p_{k+1} \nonumber
\end{align}
At $r_{k+1}$, the remaining processing time for $J_k$ is $p_k - w \leq p_{k+1}$. Therefore, SRPT completes the execution of $J_k$ before scheduling $J_{k+1}$.
\end{itemize} \ \\
\item On another hand, consider the case where the optimal schedule is \emph{not compact}. 
Let $t$ be the first moment when there is an idle time in the optimal schedule between two consecutive jobs $J_k$ and $J_{k+1}$ even though some job $J_l$ is available. 
If $w$ is length of idle time: $w = S_{k+1} - C_{k}$, then from previous case, it follows that $w<p_l - p_k$. 
Therefore, during this idle time, SRPT runs job $J_l$ such that remaining processing time of $J_l$ at $t+w$ is $p_l - w > p_k$. 
Hence, SRPT preempts job $J_l$ at $t+w$, and schedules $J_k$. 
Therefore at time $t+w$, SRPT and optimal have same set of uncompleted jobs. But the amount of work left in SRPT is less than or equal to the amount of job left in an optimal schedule. 
\end{itemize}
Iteratively applying these above arguments at each idle period, it follows that SRPT is a lower bound for optimal schedule for non-preemptive total stretch.
\end{proof}
\section{Proof of theorem~\ref{final-theorem}} \label{final-proof}
\begin{proof}
Using the bound as stated in the corollary~\ref{opt-lb} and replacing $\sum\limits^{sptm}$ with $\sum\limits^{opt}_m$ in corollary~\ref{spt-sptm}, we get:
\begin{align*}
\sum\limits_{m}^{spt} &\leq (\Delta - \frac{1}{\Delta} + 1) (\sum_{m}^{opt} - \frac{1}{2}(1 - \frac{1}{m})n) + (1-\frac{1}{m})n \\
&\leq  (\Delta - \frac{1}{\Delta} + 1) \sum_{m}^{opt} + \frac{1}{2}(1-\frac{1}{m})n - \frac{\Delta - \frac{1}{\Delta} + 1}{2} (1- \frac{1}{m})n \\
&\leq (\Delta - \frac{1}{\Delta} + 1) \sum_{m}^{opt} + \frac{1}{2}(1-\frac{1}{m})n 
\end{align*}
Now combining the fact that $\forall j \in J, s_j \geq 1$ for all schedules. We get $n \leq \sum\limits^{opt}_m$. Replacing this inequality in above inequality, we obtain:
\begin{align*}
\sum\limits_{m}^{spt} &\leq (\Delta - \frac{1}{\Delta} + 1 + \frac{1}{2}(1 - \frac{1}{m}) \sum_{m}^{opt} \\
&\leq (\Delta - \frac{1}{\Delta} + \frac{3}{2} - \frac{1}{2m}) \sum_{m}^{opt}
\end{align*}
\end{proof}
\end{document}